\documentclass[twocolumn,amsmath,amssymb,10pt,aps]{revtex4}
  
\usepackage[utf8]{inputenc}
\usepackage[T1]{fontenc}
\usepackage{amsmath}
\usepackage{amsfonts}
\usepackage{graphicx}
\usepackage{yfonts}
\usepackage{color}
\usepackage[normalem]{ulem}
\usepackage{amsthm}
\usepackage{bm}
\usepackage{bbm}
\usepackage{mathtools}
\usepackage{array}

\def\<{\langle}
\def\>{\rangle}
\usepackage{placeins}
\usepackage{enumitem}

\newcommand{\Tr}{\mathrm{Tr}}
\def\oper{{\mathchoice{\rm 1\mskip-4mu l}{\rm 1\mskip-4mu l}
{\rm 1\mskip-4.5mu l}{\rm 1\mskip-5mu l}}}
\newtheorem{Theorem}{Theorem}

\newtheorem{Corollary}{Corollary}

\newtheorem{Remark}{Remark}
\newtheorem{Proposition}{Proposition}

\begin{document}

\title{\bf Regularized maximal fidelity of the generalized Pauli channels}
\author{Katarzyna Siudzi{\'n}ska}
\affiliation{ Institute of Physics, Faculty of Physics, Astronomy and Informatics \\  Nicolaus Copernicus University,
Grudzi\k{a}dzka 5/7, 87--100 Toru{\'n}, Poland}

\begin{abstract}
We consider the asymptotic regularization of the maximal fidelity for the generalized Pauli channels, which is a problem similar to the classical channel capacity. In particular, we find the formulas for the extremal channel fidelities and the maximal output $\infty$-norm. For wide classes of channels, we show that these quantities are multiplicative. Finally, we find the regularized maximal fidelity for the channels satisfying the time-local master eqeuations.
\end{abstract}

\maketitle

\section{Introduction}

The noisiness of a quantum channel is characterized by its capacity, which is the probability that the transmitted information does not get distorted. Finding the channel capacity is still an open problem, even in the simplest case of the classical capacity
\begin{equation}
C(\Lambda)=\lim_{n\to\infty}\frac 1n \chi(\Lambda^{\bigotimes n}),
\end{equation}
given in terms of the Holevo capacity $\chi$ \cite{Holevo,sw}. This is the case because calculating the Holevo capacity requires finding the maximal value of the entropic expression over ensembles of states and their probabilities of occurence. The problem of finding the classical capacity simplifies significantly for unitarily covariant quantum channels, for which $C(\Lambda)=\chi(\Lambda)$ as long as the minimal output entropy $S_{\min}(\Lambda)$ is additive \cite{Holevo2,mds}.

Recently, Ernst and Klesse \cite{Klesse} have proposed a toy problem that is {\it structurally similar but technically far less demanding} than finding the classical capacity. Namely, their goal is to determine the asymptotic regularization
\begin{equation}
f_{\max}^{(\infty)}(\Lambda)=\lim_{n\to\infty}\sqrt[n]{f_{\max}(\Lambda^{\bigotimes n})}
\end{equation}
of the maximal channel fidelity $f_{\max}$, which measures the distortion of states under the action of the channel. It turns out that there is a relation between $f_{\max}^{(\infty)}(\Lambda)$ and the maximal output $\infty$-norm $\nu_\infty(\Lambda)$, which is a measure of the optimal output purity. Moreover, if the maximal output $\infty$-norm is multiplicative, then $f_{\max}^{(\infty)}(\Lambda)=\nu_\infty(\Lambda)$.

Our goal is to analyze the regularized maximal fidelity for the generalized Pauli channels. For these channels, we find the exact formulas for the extremal values of the channel fidelity and the maximal output $\infty$-norm. We show that $\nu_\infty$ is indeed reached on the projectors onto the mutually unbiased bases, which was conjectured by Nathanson and Ruskai \cite{Ruskai}. Next, we prove the multiplicativity of $\nu_\infty$ for certain families of channels. Finally, we find the regularized maximal fidelity for a wide class of the generalized Pauli channels.

\section{Generalized Pauli channels}

Consider the $d$-dimensional Hilbert space $\mathcal{H}$ with the maximal number $N(d)=d+1$ of mutually unbiased bases $\{\psi_0^{(\alpha)},\dots,\psi_{d-1}^{(\alpha)}\}$ \cite{Wootters,MAX}. One defines the generalized Pauli channels $\Lambda:\mathcal{B}(\mathcal{H})\to\mathcal{B}(\mathcal{H})$ \cite{Ruskai,mub_final} by
\begin{equation}\label{GPC}
\Lambda=\frac{dp_0-1}{d-1}\oper+\frac{d}{d-1}\sum_{\alpha=1}^{d+1}p_\alpha\Phi_\alpha,
\end{equation}
where $p_\alpha$ denotes the probability distribution,
\begin{equation}
\Phi_\alpha[X]=\sum_{k=0}^{d-1}P_k^{(\alpha)}XP_k^{(\alpha)},
\end{equation}
and $P_k^{(\alpha)}:=|\psi_k^{(\alpha)}\>\<\psi_k^{(\alpha)}|$ is the rank-1 projector. For $d=2$, eq. (\ref{GPC}) reduces to the Pauli channel
\begin{equation}\label{PC}
\Lambda=\sum_{\alpha=0}^{3}p_\alpha\sigma_\alpha\rho\sigma_\alpha,
\end{equation}
where $\sigma_0=\mathbb{I},\ \sigma_1,\ \sigma_2,\ \sigma_3$ are the Pauli matrices.

The generalized Pauli channels satisfy the eigenvalue equations
\begin{equation}\label{GPC_eigenvalue_eq}
\Lambda[U_\alpha^k]=\lambda_\alpha U_\alpha^k,\qquad k=1,\ldots,d-1,
\end{equation}
and $\Lambda[\mathbb{I}]=\mathbb{I}$. Interestingly, the eigenvectors are the unitary operators constructed from the projectors onto the mutually unbiased bases,
\begin{equation}\label{U}
U_{\alpha}^k=\sum_{l=0}^{d-1}\omega^{kl}P_l^{(\alpha)},
\end{equation}
with $\omega = e^{2\pi i/d}$. 
The eigenvalues $\lambda_\alpha$ are $(d-1)$-times degenerate, and they are related to the probability distribution via
\begin{equation}\label{GPC_eigenvalues}
\lambda_\alpha=\frac{1}{d-1}\left[d(p_0+p_\alpha)-1\right].
\end{equation}
The inverse relation reads
\begin{equation}\label{c1}
p_0=\frac{1}{d^2}\left[1+(d-1)\sum_{\alpha=1}^{d+1}\lambda_\alpha\right],
\end{equation}
\begin{equation}\label{c2}
p_\alpha=\frac{d-1}{d^2}\left[1+d\lambda_\alpha-\sum_{\beta=1}^{d+1} \lambda_\beta\right].
\end{equation}
The necessary and sufficient conditions for the generalized Pauli channel to be a completely positive and trace-preserving map are the generalized Fujiwara-Algoet conditions \cite{Fujiwara, Ruskai, Zyczkowski}
\begin{equation}\label{Fuji-d}
-\frac{1}{d-1}\leq\sum_{\beta=1}^{d+1}\lambda_\beta\leq 1+d\min_{\beta}\lambda_\beta.
\end{equation}

The generalized Pauli channels were first considered by Nathanson and Ruskai \cite{Ruskai} as the {\it Pauli diagonal channels constant on axes}. Ohno and Petz \cite{Petz} analyzed even more general channels, of which the generalized Pauli channels are the special case with the commutative subalgebras $\{\mathbb{I},U_\alpha^k\ |\ k=1,\dots,d-1\}$. Their applications range between the quantum process tomography \cite{QPT}, optimal parameter estimation \cite{Ruppert}, and geometrical quantum mechanics \cite{geom}. In the theory of open quantum systems and non-Markovian dynamics, the channels and their evolution were analyzed in both the time-local evolution \cite{mub_final,ICQC} and the memory kernel approach \cite{memory_final}. In the present paper, we focus on the other properties, like state distortion and purity.

\section{Channel fidelity}

The fidelity is the measure of distance between two quantum states \cite{Nielsen, Zyczkowski}. It helps us to determine how distinguished the states are from one another. Uhlmann \cite{Uhlmann} defined the fidelity between the density operators $\rho$ and $\sigma$ by
\begin{equation}\label{statesfidelity}
F(\rho,\sigma):=\left(\Tr\sqrt{\sqrt{\rho}\sigma\sqrt{\rho}}\right)^2,
\end{equation}
where $0\leq F(\rho,\sigma)\leq 1$ and $F(\rho,\sigma)=1$ if and only if $\rho=\sigma$. On the basis of this simple formula, many other types of fidelity were derived, like the entanglement fidelity \cite{schumacher96}, average fidelity \cite{nielsen02}, or regularized maximum pure state input-output fidelity \cite{Klesse}. The channel-state duality \cite{jam72} allows one to introduce the notion of channel fidelity $F(\rho,\Lambda[\rho])$, which measures the fidelity between the input $\rho$ and output $\Lambda[\rho]$ states \cite{raginsky01}. Because of its convexity, the extremal values of $F(\rho,\Lambda[\rho])$ are reached at pure states represented by rank-1 projectors $P$. Hence, one defines the minimal and maximal channel fidelity \cite{Zycz} by
\begin{equation}\label{channelfidelity}
\begin{split}
f_{\min}(\Lambda)=\min_PF(P,\Lambda[P])=\min_P \Tr(P\Lambda[P]),\\
f_{\max}(\Lambda)=\max_PF(P,\Lambda[P])=\max_P \Tr(P\Lambda[P]).
\end{split}
\end{equation}
They measure how much a given quantum channel can distort an initial state. The more $\Lambda$ resembles the identity map $\oper$, the less $\rho$ changes under a single action of the channel.

\begin{Theorem}\label{THM}
For the generalized Pauli channel $\Lambda$ defined by eq. (\ref{GPC}), the minimal and maximal channel fidelities are equal to
\begin{align}
&f_{\min}(\Lambda)=\frac{1}{d}\left[1+(d-1)\lambda_{\min}\right],\label{A1}\\
&f_{\max}(\Lambda)=\frac{1}{d}\left[1+(d-1)\lambda_{\max}\right],\label{A2}
\end{align}
where $\lambda_{\max}=\max_\alpha\lambda_\alpha$ and $\lambda_{\min}=\min_\alpha\lambda_\alpha$.
\end{Theorem}

\begin{proof}
To calculate the channel fidelity, we need to know how $\Lambda$ transforms pure initial states. Any rank-1 projector $P$ can be written in the unitary basis $\{\mathbb{I},U_\alpha^k\}$ introduced in eq. (\ref{U}). Namely,
\begin{equation}\label{P}
P=\frac 1d \left(\mathbb{I}+\sum_{\alpha=1}^{d+1}\sum_{k=1}^{d-1}x_{\alpha k}U_\alpha^k\right),
\end{equation}
where $x_{\alpha k}$ are complex parameters. Now, we find that
\begin{equation}\label{prod}
\Lambda[P]=\frac 1d \left[\mathbb{I}+\sum_{\alpha=1}^{d+1}\sum_{k=1}^{d-1}
\lambda_\alpha x_{\alpha k}U_\alpha^k\right],
\end{equation}
which allows us to obtain the channel fidelity for the generalized Pauli channels,
\begin{equation}\label{sch}
F(P,\Lambda[P])=\Tr(P\Lambda[P])=\frac{1}{d}\left(1+\sum_{\alpha=1}^{d+1}\lambda_\alpha
\sum_{k=1}^{d-1}|x_{\alpha k}|^2\right).
\end{equation}
Recall that $P$ is a rank-1 projector, and hence
\begin{equation}
\Tr P^2=\frac{1}{d}\left(1+\sum_{\alpha=1}^{d+1}\sum_{k=1}^{d-1}
|x_{\alpha k}|^2\right)=1.
\end{equation}
The above condition is equivalent to
\begin{equation}
\sum_{\alpha=1}^{d+1}\sum_{k=1}^{d-1}|x_{\alpha k}|^2=d-1.
\end{equation}
Therefore, $F(P,\Lambda[P])$ reaches its minimal and maximal values if the only non-vanishing coefficients are $x_{\alpha_\ast k}$ and $x_{\alpha_\# k}$, respectively, where $\lambda_{\alpha_\ast}=\lambda_{\min}$ and $\lambda_{\alpha_\#}=\lambda_{\max}$. The minimal and maximal channel fidelities are reached at the projectors onto the mutually unbiased bases,
\begin{align}
&f_{\min}(\Lambda)=F(P_k^{(\alpha_\ast)},\Lambda[P_k^{(\alpha_\ast)}]),\\
&f_{\max}(\Lambda)=F(P_k^{(\alpha_\#)},\Lambda[P_k^{(\alpha_\#)}]),
\end{align}
where, from eq. (\ref{U}),
\begin{equation}
P_k^{(\alpha)}=\frac 1d 
\left(\mathbb{I}+\sum_{l=1}^{d-1}\omega^{-kl}U_\alpha^l\right).
\end{equation}
\end{proof}

\begin{Remark}
The minimal and maximal channel fidelities from Theorem \ref{THM} can be equivalently written in terms of the probability distribution $p_\alpha$ as
\begin{align}
&f_{\min}(\Lambda)=p_0+\min_{\alpha>0} p_\alpha,\\
&f_{\max}(\Lambda)=p_0+\max_{\alpha>0} p_\alpha.
\end{align}
\end{Remark}

\section{Maximal output $p$-norm}

Quantum channels $\Lambda$ transform input states $\rho$ into output states $\Lambda[\rho]$. While it is relatively easy to control the input, the attainable output states depend on the channel's properties. In particular, it is not always possible to find such $\rho$ for which $\Lambda[\rho]$ is pure. In these cases, one can ask how close the outputs are to pure states. This is measured by optimal output purity measures. One such measure is the maximal output $p$-norm defined as follows,
\begin{equation}
\nu_p(\Lambda):=\max_P||\Lambda[P]||_p,
\end{equation}
where the Schlatten $p$-norm reads
\begin{align}
&||\Lambda[P]||_p:=(\Tr\Lambda[P]^p)^{1/p},\qquad 1\leq p<\infty,\\
&||\Lambda[P]||_\infty:=\max_Q\Tr(Q\Lambda[P]),
\end{align}
and $Q$ is a rank-1 projector. For product channels, it is known that
\begin{equation}
\nu_p(\Lambda\otimes\Phi)\geq\nu_p(\Lambda)\nu_p(\Phi).
\end{equation}
The maximal output $p$-norm is multiplicative if
\begin{equation}\label{mul}
\nu_p(\Lambda\otimes\Phi)=\nu_p(\Lambda)\nu_p(\Phi).
\end{equation}
Fukuda \cite{Fukuda} proved that if eq. (\ref{mul}) is satisfied for all $\Phi:\mathbb{I}_{d_1}/d_1\longmapsto\mathbb{I}_{d_2}/d_2$, then it holds for any $\Phi$.

Nathanson and Ruskai \cite{Ruskai} derived the exact formula for the maximal output $2$-norm of the generalized Pauli channel,
\begin{equation}
\nu_2(\Lambda)=\sqrt{\frac 1d \left[1+(d-1)\max_\alpha\lambda_\alpha^2\right]}.
\end{equation}
They also conjectured that the maximal output $p$-norm is achieved on the projectors onto the mutually unbiased bases $P_k^{(\alpha)}$ and proved it in two special cases: $p=2$ and $p=\infty$. Indeed, the maximal value of $||\Lambda[P]||_2$ is reached at $P_k^{(\alpha_0)}$, where $\alpha_0$ numbers the eigenvalue whose module is maximal. Moreover, the maximal output $2$-norm is related to the maximal channel fidelity in the following way,
\begin{equation}\label{fmaxnu2}
f_{\max}(\Lambda^\dagger\Lambda)=\nu_2^2(\Lambda).
\end{equation}
Finally, Nathanson and Ruskai \cite{Ruskai} showed that $\nu_2(\Lambda\otimes\Phi)$ is multiplicative for the generalized Pauli channels $\Lambda$ (with arbitrary $\Phi$).

\begin{Remark}\label{RR}
The maximal output $2$-norm $\nu_2(\Lambda)$ and the maximal channel fidelity $f_{\max}(\Lambda)$ are attained at the same state $P_k^{(\alpha_{\#})}$ if and only if $|\lambda_{\max}|\geq|\lambda_{\min}|$, where $\lambda_{\alpha_{\#}}=\max_\alpha \lambda_\alpha$.
\end{Remark}

\begin{Remark}\label{RR2}
The maximal output $2$-norm $\nu_2(\Lambda)$ and the minimal channel fidelity $f_{\min}(\Lambda)$ are reached at the same state $P_k^{(\alpha_{\ast})}$ if and only if $\lambda_{\max}^2\leq\lambda_{\min}^2$, where $\lambda_{\alpha_{\ast}}=\min_\alpha \lambda_\alpha$.
\end{Remark}

From the multiplicativity of $\nu_2(\Lambda)$, it follows that the extremal channel fidelities can be multiplicative, as well.

\begin{Proposition}\label{PP}
The minimal and maximal channel fidelity for the generalized Pauli channels are multiplicative in the sense that
\begin{enumerate}[label={(\roman*)}]
\item if $|\lambda_{\max}|\leq|\lambda_{\min}|$, then
\begin{equation}
f_{\min}(\Lambda\otimes\Lambda)=f_{\min}^2(\Lambda);
\end{equation}
\item if $|\lambda_{\max}|\geq|\lambda_{\min}|$, then
\begin{equation}
f_{\max}(\Lambda\otimes\Lambda)=f_{\max}^2(\Lambda).
\end{equation}
\end{enumerate}
\end{Proposition}

\begin{proof}
The above equalities follow directly from eqs. (\ref{fmaxnu2}) and (\ref{mul}). For {\it (i)} with $|\lambda_{\max}|\geq|\lambda_{\min}|$, one has
\begin{equation}
\begin{split}
f_{\max}(\Lambda^\dagger\Lambda)&=\nu_2^2(\Lambda)=\nu_2(\Lambda\otimes\Lambda)\\&
=\sqrt{f_{\max}(\Lambda^\dagger\Lambda\otimes\Lambda^\dagger\Lambda)}.
\end{split}
\end{equation}
For {\it (2)}, the proof is analogical.
\end{proof}

Now, we derive the formula for the maximal output $\infty$-norm.

\begin{Proposition}
For the generalized Pauli channel $\Lambda$, the maximal output $\infty$-norm is given by
\begin{equation}\label{infnorm}
\begin{split}
\nu_\infty(\Lambda)&=\max_{P,Q}\Tr(Q\Lambda[P])\\&=\frac 1d
\max\left\{1+(d-1)\lambda_{\max},1-\lambda_{\min}\right\},
\end{split}
\end{equation}
where $\lambda_{\max}=\max_\alpha\lambda_\alpha$ and $\lambda_{\min}=\min_\alpha\lambda_\alpha$.
\end{Proposition}

\begin{proof}
To calculate $\Tr(Q\Lambda[P])$ for the generalized Pauli channel $\Lambda$, let us parametrize the projectors $P$, $Q$ by (\ref{P}) and
\begin{equation}
Q=\frac 1d \left(\mathbb{I}+\sum_{\alpha=1}^{d+1}\sum_{k=1}^{d-1}y_{\alpha k}U_\alpha^k\right).
\end{equation}
Note that the condition $0\leq\Tr PQ\leq 1$ is equivalent to
\begin{equation}\label{warunek}
-1\leq\sum_{\alpha=1}^{d+1}\sum_{k=1}^{d-1}
x_{\alpha k}\overline{y}_{\alpha k}\leq d-1.
\end{equation}
Moreover, from eq. (\ref{prod}), it follows that
\begin{equation}\label{wynik}
\Tr(Q\Lambda[P])=\frac{1}{d}\left[1+\sum_{\alpha=1}^{d+1}\sum_{k=1}^{d-1}
\lambda_\alpha x_{\alpha k}\overline{y}_{\alpha k}\right].
\end{equation}
The maximal value of the above quantity is reached when $x_{\alpha k}=y_{\alpha k}=0$ for $\alpha\neq\alpha_\ast$.
\begin{enumerate}[label={(\roman*)}]
\item If $\lambda_{\alpha_\ast}\geq 0$, then the maximal value
\begin{equation}
\max_{P,Q}\Tr(Q\Lambda[P])=\frac{1}{d}\left[1+\lambda_{\alpha_\ast}\sum_{k=1}^{d-1}
 x_{\alpha_\ast k}\overline{y}_{\alpha_\ast k}\right]
\end{equation}
of eq. (\ref{wynik}) is attained for the maximal bound of (\ref{warunek}),
\begin{equation}
\sum_{k=1}^{d-1}
\lambda_{\alpha_\ast} x_{\alpha_\ast k}\overline{y}_{\alpha_\ast k}=d-1,
\end{equation}
and $\lambda_{\alpha_\ast}=\lambda_{\max}$.
In this case,
\begin{equation}
\max_{P,Q}\Tr(Q\Lambda[P])=\frac{1}{d}\left[1+(d-1)\lambda_{\max}\right].
\end{equation}
From the results for $2$-norms, we know that the above formula corresponds to the choice
\begin{equation}
Q=P=P_k^{(\alpha_\ast)}=\frac 1d 
\left(\mathbb{I}+\sum_{l=1}^{d-1}\omega^{-kl}U_{\alpha_\ast}^k\right).
\end{equation}
\item If $\lambda_{\alpha_\ast}\leq 0$, then the maximal value
\begin{equation}
\max_{P,Q}\Tr(Q\Lambda[P])=\frac{1}{d}\left[1-|\lambda_{\alpha_\ast}|\sum_{k=1}^{d-1}
 x_{\alpha_\ast k}\overline{y}_{\alpha_\ast k}\right]
\end{equation}
of eq. (\ref{wynik}) is attained for the minimal bound of (\ref{warunek}),
\begin{equation}
\sum_{k=1}^{d-1}
\lambda_{\alpha_\ast} x_{\alpha_\ast k}\overline{y}_{\alpha_\ast k}=-1,
\end{equation}
and $\lambda_{\alpha_\ast}=\lambda_{\min}$.
In this case,
\begin{equation}
\max_{P,Q}\Tr(Q\Lambda[P])=\frac{1}{d}\left[1-\lambda_{\min}\right].
\end{equation}
\end{enumerate}
The projectors $P$ and $Q$ that maximize eq. (\ref{wynik}) are as follows,
\begin{align}
P=P_k^{(\alpha_\ast)}=\frac 1d 
\left(\mathbb{I}+\sum_{l=1}^{d-1}\omega^{-kl}U_{\alpha_\ast}^k\right),\\
Q=P_m^{(\alpha_\ast)}=\frac 1d 
\left(\mathbb{I}+\sum_{l=1}^{d-1}\omega^{-ml}U_{\alpha_\ast}^m\right),
\end{align}
where $k\neq m$.
\end{proof}

Let us compare our results with the analysis done for the Pauli channels in \cite{Klesse}. For $d=2$, eq. (\ref{infnorm}) produces
\begin{equation}
\nu_\infty(\Lambda)=\max_{P,Q}\Tr(Q\Lambda[P])=\frac 12\max\left\{1+\lambda_{\max},1-\lambda_{\min}\right\}.
\end{equation}
In terms of the probability distribution,
\begin{equation}
\nu_\infty(\Lambda)=\frac 12(1+\lambda_{\max})=p_0+p_{\max},
\end{equation}
provided that
$\lambda_{\max}+\lambda_{\min}\geq 0$. Now, $\Tr(Q\Lambda[P])$ reaches the maximal value for $P=Q=P_0^{(\alpha_\ast)}$ or $P=Q=P_1^{(\alpha_\ast)}$, where $\lambda_{\max}=\lambda_{\alpha_\ast}$. Note that
\begin{equation}
\lambda_{\max}+\lambda_{\min}=p_0-p_{\mathrm{mid}}\geq 0,
\end{equation}
where $p_{\min}\leq p_{\mathrm{mid}}\leq p_{\max}$ and $\{p_{\min},p_{\mathrm{mid}},p_{\max}\}=\{p_1,p_2,p_3\}$. On the other hand,
\begin{equation}
\nu_\infty(\Lambda)=\frac 12(1-\lambda_{\min})=p_{\mathrm{mid}}+p_{\max}
\end{equation}
only when $\lambda_{\max}+\lambda_{\min}\leq 0$, i.e. $p_0-p_{\mathrm{mid}}\leq 0$. This time, the maximum of $\Tr(Q\Lambda[P])$ is reached for $\{P=P_0^{(\alpha_\ast)},Q=P_1^{(\alpha_\ast)}\}$ or $\{P=P_1^{(\alpha_\ast)},Q=P_0^{(\alpha_\ast)}\}$, where $\lambda_{\min}=\lambda_{\alpha_\ast}$. These results coincide with the maximal output $\infty$-norms for the Pauli channels found in \cite{Klesse}.

\section{Asymptotic regularizations of the channel fidelity}

After Ernst and Klesse \cite{Klesse}, we introduce the $n$-th regularization of the extremal channel fidelities and the maximal output $p$-norm,
\begin{equation}
f_{\min}^{(n)}(\Lambda)=\sqrt[n]{f_{\min}(\Lambda^{\bigotimes n})},
\end{equation}
\begin{equation}
f_{\max}^{(n)}(\Lambda)=\sqrt[n]{f_{\max}(\Lambda^{\bigotimes n})},
\end{equation}
\begin{equation}
\nu_p^{(n)}(\Lambda)=\sqrt[n]{\nu_p(\Lambda^{\bigotimes n})}.
\end{equation}
In particular, one talks about asymptotic regularizations if $n=\infty$.
The authors consider the asymptotic regularization of the maximal channel fidelity as a {\it toy model} in the channel capacity problem. They prove the following relation,
\begin{equation}
f_{\max}(\Lambda)\leq\nu_\infty(\Lambda)\leq \nu_\infty^{(\infty)}(\Lambda)=f_{\max}^{(\infty)}(\Lambda).
\end{equation}

\begin{Corollary}\label{C1}
If $|\lambda_{\max}|\geq|\lambda_{\min}|$, then $f_{\max}^{(n)}(\Lambda)$ is multiplicative, and hence
\begin{equation}
f_{\max}^{(n)}(\Lambda)=f_{\max}(\Lambda),\qquad n=1,2,\ldots,\infty.
\end{equation}
In particular,
\begin{equation}\label{equal}
f_{\max}(\Lambda)=\nu_\infty(\Lambda)= \nu_\infty^{(\infty)}(\Lambda)=f_{\max}^{(\infty)}(\Lambda).
\end{equation}
\end{Corollary}

\begin{Corollary}
If $|\lambda_{\max}|\leq|\lambda_{\min}|$, then $f_{\min}^{(n)}(\Lambda)$ is multiplicative, and therefore
\begin{equation}
f_{\min}^{(n)}(\Lambda)=f_{\min}(\Lambda),\qquad n=1,2,\ldots,\infty.
\end{equation}
\end{Corollary}

\begin{Corollary}\label{C3}
For $\lambda_{\max}\geq-\frac{1}{d-1}\lambda_{\min}$, the maximal output $\infty$-norm and the maximal channel fidelity coincide, $\nu_\infty(\Lambda)=f_{\max}(\Lambda)$. If additionally $|\lambda_{\max}|\geq|\lambda_{\min}|$, then $\nu_\infty(\Lambda)$ is multiplicative.
\end{Corollary}

In the theory of open quantum systems, the evolution is given by dynamical maps $\Lambda(t)$ -- that is, the families of quantum channels parametrized by time $t\geq 0$. Assume that $\Lambda(t)$ is the solution of the master equation
\begin{equation}\label{me}
\dot{\Lambda}(t)=\mathcal{L}(t)\Lambda(t),\qquad \Lambda(0)=\oper,
\end{equation}
with the time-local generator $\mathcal{L}(t)$. An important property of such a dynamical map is that its eigenvalues are non-negative, $\lambda_\alpha(t)\geq 0$. It is easy to check that they satisfy the inequalities in Corollaries \ref{C1} and \ref{C3}. Therefore, the generalized Pauli channels being the solutions of eq. (\ref{me}) are good examples of the quantum channels for which the maximal channel fidelity and the maximal output $\infty$-norm are multiplicative.

\section{Conclusions}

We analyzed the channel fidelity for the generalized Pauli channels, which is the measure of distortion between input and output states. We found general analytical formulas for the minimal and maximal channel fidelity and showed that these quantities satisfy the multiplicativity conjecture for some classes of the generalized Pauli channels. Next, we focused our attention on the matter of purity of the output states, which was measured by the maximal output $p$-norm. For $p=\infty$, we derived the exact formula for the maximal output norm, and we also showed that it can be multiplicative. Finally, we analyzed the regularized maximal channel fidelity, which is technically simpler than the channel capacity problem. It turns out that our results lead to interesting implications in the theory of open quantum systems. Namely, if the generalized Pauli channel is generated by the time-local generator, then its maximal channel fidelity and maximal output $\infty$-norm are both multiplicative.

Many questions are still remaining unanswered. It would be interesting to determine whether the extremal channel fidelities and the maximal output $\infty$-norm are multiplicative for the whole spectrum of eigenvalues $\lambda_\alpha$. Also, the behaviour of the regularized channel fidelities for $n<\infty$ requires further studies.

\section*{Acknowledgements} This paper was supported by the National Science Centre project 2015/17/B/ST2/02026.


\begin{thebibliography}{99}

\bibitem{Holevo} A. S. Holevo, IEEE Trans. Info. Theor. {\bf 44}, 269-273 (1998).

\bibitem{sw} B. Schumacher and M. D. Westmoreland, Phys. Rev. A {\bf 56}, 131-138 (1997).

\bibitem{Holevo2} A. S. Holevo, arXiv:quant-ph/0212025.

\bibitem{mds} M. Mozrzymas, M. Studzi{\'n}ski, and N. Datta, J. Math. Phys. {\bf 58}, 052204 (2017).

\bibitem{Klesse} M. F. Ernst, R. Klesse, Phys. Rev. A {\bf 96}, 062319 (2017).

\bibitem{Ruskai} M. Nathanson and M. B. Ruskai, J. Phys. A: Math. Theor. {\bf 40} 8171 (2007).

\bibitem{Wootters} W. K. Wootters and B. D. Fields, Ann. Phys. {\bf 191}, 363 (1989).

\bibitem{MAX} S. Bandyopadhyay, P. Boykin, V. Roychowdhury, and F. Vatan, Algorithmica {\bf 34}, 512 (2002).

\bibitem{mub_final} D. Chru{\'s}ci{\'n}ski and K. Siudzi{\'n}ska, Phys. Rev. A {\bf 94}, 022118 (2016).

\bibitem{Fujiwara} A. Fujiwara and P. Algoet, Phys. Rev. A {\bf 59}, 3290 (1999).

\bibitem{Zyczkowski} I. Bengtsson and K. \.Zyczkowski, CUP, Cambridge 2006.

\bibitem{Petz} D. Petz and H. Ohno, Acta Math. Hungar. {\bf 124}, 165 (2009).

\bibitem{QPT} L. Ruppert and A. Magyar, in {\it Proceedings of the 9th International PhD Workshop on Systems and Control}, Izola, Slovenia, 2008.

\bibitem{Ruppert} L. Ruppert, D. Virosztek, and K. Hangos, J. Phys. A: Math. Theor. {\bf 45}, 265305 (2012).

\bibitem{geom} K. Siudzi{\'n}ska, Rep. Math. Phys. {\bf 80}, 361-372 (2017).

\bibitem{ICQC} K. Siudzi{\'n}ska and D. Chru{\'s}ci{\'n}ski, J. Math. Phys. {\bf 59}, 033508 (2018).

\bibitem{memory_final} K. Siudzi{\'n}ska and D. Chru{\'s}ci{\'n}ski, Phys. Rev. A {\bf 96}, 022129 (2017).

\bibitem{Nielsen} M. A. Nielsen and I. L. Chuang, CUP, Cambridge 2000.

\bibitem{Uhlmann}  A. Uhlmann, Rep. Math. Phys. \textbf{24}, 229 (1986).

\bibitem{schumacher96} B. Schumacher, Phys Rev A. \textbf{54}, 2614 (1996).

\bibitem{nielsen02} M. A. Nielsen, Phys. Lett. A \textbf{303}, 249 (2002).

\bibitem{jam72} A. Jamio\l{}kowski, Rep. Math. Phys. \textbf{3}, 275 (1972).

\bibitem{raginsky01} M. Raginsky, Phys. Lett. A \textbf{290}, 11 (2001).

\bibitem{Zycz} K. Zyczkowski, H.-J. Sommers, Phys. Rev. A {\bf 71}, 032313 (2005).

\bibitem{Fukuda} M. Fukuda, Quant. Inf. Proc. {\bf 6}, 179-186 (2007).


\end{thebibliography}
\end{document}